\documentclass[twoside]{article}

\usepackage[accepted]{aistats2019arxiv}
\usepackage{amssymb}
\usepackage{amsmath}
\usepackage{amsthm}
\usepackage{algorithm}
\usepackage[noend]{algpseudocode}
\usepackage{graphicx}

\usepackage{apalike}
\newtheorem{theorem}{Theorem}

\renewcommand\refname{}

\begin{document}

%

%

\twocolumn[

\aistatstitle{Effective Resistance-based Germination of Seed Sets for Community Detection}

\aistatsauthor{ Jonathan Eskreis-Winkler \And Risi Kondor}
\aistatsaddress{University of Chicago \And  University of Chicago} ]
\begin{abstract}

Community detection is, at its core, an attempt to attach an interpretable function to an otherwise indecipherable form. The importance of labeling communities has obvious implications for identifying clusters in social networks, but it has a number of equally relevant applications in product recommendations, biological systems, and many forms of classification. The local variety of community detection starts with a small set of labeled ``seed nodes," and aims to estimate the community containing these nodes. One of the most ubiquitous methods - based on its simplicity and efficiency - is personalized PageRank. The most obvious bottleneck for deploying this form of PageRank successfully is the quality of the seeds. We introduce a ``germination" stage for these seeds, where an effective resistance-based approach is used to increase the quality and number of seeds from which a community is detected. By breaking seed set expansion into a two-step process, we aim to utilize two distinct random walk-based approaches in the regimes in which they excel. In synthetic and real network data, a simple, greedy algorithm which minimizes the effective resistance diameter combined with PageRank achieves clear improvements in precision and recall over a standalone PageRank procedure.

\end{abstract}

\section{INTRODUCTION}

Nascent abilities to collect progressively more massive amounts of data continually motivate new approaches to large-scale data analysis. Recognizing that many data collection pursuits involve unstructured data, and that the data collected is of a scale so large that it is infeasible for a human being to label training data with their associated classes, the fields of unsupervised learning and semi-supervised learning gain new importance \cite{zhou_machine_2017,banerjee_graph-based_2017}. We focus on a contemporary challenge in the class of community detection (CD) problems. The aim of CD is to leverage the known features of a complex system - namely its nodes and edges - to estimate unobserved functions, such as class labelings, on the network. The foundation of network analysis in general and CD in particular is that the nodes and edges of a graph convey valuable information, information that can be deployed for practical applications. In the last decade, the CD problem has motivated solutions in a constantly growing set of fields ripe for application \cite{gargi_large-scale_2011,shai_case_2017,garcia_applications_2018}.

A proliferation of work in CD has given rise to the formulation of several distinct problems involving the detection of communities that have different starting assumptions and objectives. Namely, given a graph $G = (V,E)$, one type of community detection aims to develop a full partition of the nodes $V$ into non-overlapping groups \cite{andersen_local_2006,fortunato_community_2010}. Another type acknowledges that networks will not always admit partitions into disjoint subgraphs and so its objective is to detect possibly overlapping communities, maximizing graph coverage \cite{palla_uncovering_2005,xie_overlapping_2013}. A less ambitious type is the detection of a single community. This problem begins with a predefined set of seeds which are known to be members of the sought-after community \cite{mehler_expanding_2009,yang_defining_2015,hollocou_multiple_2018,van_laarhoven_generative_2018}. 

The term ``seed set expansion" can refer to two scenarios. The first scenario is an unsupervised problem that demands a graph partition and might be addressed with a two-step algorithm: it would first identify appropriate seeds and then use them to extrapolate towards a  partition of the graph into possibly overlapping communities \cite{whang_overlapping_2013,moradi_local_2014}. The second scenario is the detection of a community around a predefined seed set and is equivalent to the ``less ambitious" problem described above. In both scenarios a seed set is being expanded, but while the former is an unsupervised problem, the latter is semi-supervised. Our approach is addressing this last form of CD. 

In addition to the many forms of CD, the number of scoring functions by which these methods are evaluated has also grown rapidly. Central to any approach to CD is the specification of a scoring function to assess the degree to which community structure is present. The most popular and effective scoring functions are based either on maximizing triadic closures, i.e. the local clustering coefficient, or on minimizing conductance \cite{yang_defining_2015}.

Our work focuses on improving the state of the art in CD for seed set expansion around a single node, or a small set of nodes. A dominant force in this problem setting is a CD-oriented version of PageRank deemed ``personalized PageRank" (PPR). Whereas the original PageRank algorithm is a ranking of nodes' importance (using PageRank contributions) in a general graph, PPR focuses on ranking nodes with scores based on their proximity to a pre-defined seed set \cite{bahmani_fast_2010,whang_overlapping_2013,kloumann_block_2017}. Our contribution builds upon existing applications of PageRank by introducing an effective resistance-based \textit{germination stage} before the propagation of PPR weights begins. Specifically, we propose a two-step approach using PPR on a vector which distributes weight over a new, revised seed set. This set includes the original seed nodes, but it also adds a selection of nodes that are chosen to minimize the effective resistance (ER) diameter of the revised seed set. The stopping rule for identifying the respective points at which the germination stage and PPR are stopped is based on community scoring functions. 

The remainder of the paper is organized as follows: in Section \ref{background} we review the role that random walks play in the theoretical underpinnings of PageRank and ER. The manner in which these two measures of node similarity differ is illustrated and the rationale for valuing ER in CD is presented. In Section \ref{alg}, we present the germination stage algorithm along with its time complexity. We also make explicit the connection between our algorithm and minimizing the ER diameter. In Section \ref{experiments}, we provide synthetic and empirical data experiments to show the efficacy of our two-step approach in practice. We conclude in Section \ref{conclusion} by outlining the new ideas that this work motivates and anticipating developments in the near future for CD.

\section{BACKGROUND} \label{background}

\subsection{Conventions of Notation}
For simplicity, this paper considers undirected, unweighted graphs $G = (V,E)$. The number of nodes is $|V| = n$ and the number of edges is $|E| = m$. If for an edge $e\in E$, it is the case that $e = (u,v)$ for $u,v\in V$, then $u,v,$ may be denoted as $e_+$ and $e_-$ for efficiency. When indices of a vector $v$ are denoted, the first through $k^\mathrm{th}$ indices are $v[1:k]$ or the $k^\mathrm{th}$ index is $v[k]$.

\subsection{Problem Statement}

Given a seed set $S\subset V$ and assuming some ground truth community $C\subset V$ exists such that $S\subset C$, we are interested in estimating a community $\hat C\subset V$, such that $S\subset \hat C$ and $\hat C$ minimizes a loss function measuring community structure. $S$ may be a set of a single vertex in $V$ or a collection of vertices. Whereas some CD scenarios, such as the ``overlapping communities" problem, aim to achieve a high proportion of coverage on the entire network, the aim of this seed set expansion problem is to find a single high quality community subject only to the constraint that it cover the seeds. In this sense we are addressing a one-class classification problem.

There are two components to addressing this problem. The first is to define the objective function to be minimized which quantifies the quality of community structure. The second component is to achieve this minimum, or to approximately achieve it. Take for instance the scoring function $f_c: S \subset V \rightarrow \mathbb{R}$ which measures the conductance of a subgraph. The global objective function based on conductance is:
$$C^* = \mathrm{argmin}_{C\subset V} f_c(C)$$
This problem is NP hard, and so we must find some other recourse towards identifying a viable solution. Thus begins the second component of CD, devising a procedure for actually finding a subgraph with strong community structure.

\subsection{Random Walks and Community Detection}

Random walks on graphs are a fundamental key to understanding the relatedness between vertices. In a network analysis setting, it is frequently assumed that besides the edge set there are no known vertex-associated features. So, we are left to assume that any nodes' interrelatedness is determined by whether they are edge-connected, if their shortest path distance is small, or if it is ``easy" to travel from one node to another. This underlies the popularity of random walk-based kernels in graph-based learning problems.

It is natural to consider the wealth of information that random walks can provide for any problem in CD. In the absence of a single, agreed-upon definition for community structure, the basic idea of ``more relatedness inwards than outwards" aligns well with the fact that a random walk after $k$ steps is more likely to be at any individual node in its community than it is to be at any node in any other community. There are several manners in which the role of random walks may be formalized. One popular example is based on the PageRank vector. It produces an estimate of the stationary distribution of a random walk which starts at a set of chosen nodes. If there is truth to the idea that there is significantly more interrelatedness inside a community than at its boundary, then this PageRank vector should assign nodes inside the community higher scores than those which are outside of it. A different formulation is the commute time $C_{v_iv_j}$ between two nodes. The random walk interpretation of this quantity is the expected amount of time for a random walk beginning at $v_i$ to arrive at $v_j$ and then return to $v_i$. If at any stage of a random walk it is more likely for the walk to stay inside a community than move to a different one, then commute times within a community will typically be shorter than those which must cross the community boundary. 

\subsection{Personalized PageRank}

The PageRank algorithm was originally introduced as a means of providing rankings for the pages of the World Wide Web \cite{page_pagerank_1999}. In the two decades since its introduction, it has evolved as a technique for a wide set of problems in a number of applied fields. The specific development we focus on here is personalized PageRank (PPR), a variant of the PageRank algorithm where the initial mass is fixed at a single node or subset of the nodes $V_o \subset V$ (with zero mass elsewhere) \cite{fogaras_towards_2005,kloumann_block_2017}. This translates to an initial zero vector $x \in \mathbb{R}^n$, except that $\forall k \in V_o, \ x_k = \frac{1}{|V_o|}$. The PPR vector is then calculated based on this initial state; the degree to which two nodes $v_i,v_j \not\in V_o$ are more or less similar to the seed set $V_o$ is governed by means of their respective indices in the PageRank vector $x_{v_i},x_{v_j}$. This method is equivalent to finding the steady state of a random walk with restarts whose mass is initially distributed evenly at the set of nodes which comprise the seed set. Node $v$'s score thus quantifies how likely it is that a random walk which starts at the seed set will end at node $v$.

This procedure translates into a CD scheme when the indices of the PageRank contributions are listed in a vector $x_\mathrm{PageRank}$ (after having listed the indices of the vertices in $V_o$) in the order of their score ranking \cite{andersen_communities_2006}. The next step in relating this vector to a community is minimizing a scoring function $f: W\subset V \rightarrow \mathbb{R}$ which measures the strength of community structure. The chosen community is then $x_\mathrm{PageRank}[1:k_o]$ where $k_o=\mathrm{argmin}_{1\le k \le n} f\left(x_\mathrm{PageRank}[1:k_o]\right)$. 

\subsection{Effective Resistance}

By interpreting a general graph as an electric network, one can study graph properties by considering how electricity flows throughout the network. All edges in the graph are considered resistors (in the case of an unweighted graph, they have equal resistance) and the relationship between the resistance of an edge $r_e$, the electric current over the edge $f_e$, and the electric potential difference between the nodes connected by said edge $v_{e^+} - v_{e^-}$ are governed by Ohm's Law: 
$$f_e = \frac{v_{e^+} - v_{e^-}}{r_e}$$ 
That is, for a unit of current to flow across an edge, the potential difference at the edge's endpoints must be equal to the resistance. Summarized in matrix form, where $p \in \mathbb{R}^{n}$ is a vector of potential differences across edges and $B\in \{-1,0,1\}^{m\times n}$ is an incidence matrix of $G$, the flow over every edge is $f = B p$. A second property of an electric network is the flow conservation property. If a unit of electric current is sent from node $a$ to $b$, then at any vertex $v\not= a,b$ in the graph, the sum of the flow of all edges connected to $v$ is 0. Summarized in matrix form, where $\delta_i$ is the Dirac delta function, $B^Tf = \delta_a - \delta_b$. Noting that $BB^T = L_G$, the graph Laplacian of $G$, and using Ohm's Law with the flow conservation property, we find that $$B^TB p = L_G p = \delta_a - \delta_b \Longleftrightarrow p = L_G^\dagger(\delta_a - \delta_b)$$
This last formula provides a means of deducing the potential differences across all nodes for a given electric flow across the network. For a single unit of current to flow across an edge, the resistance of the edge must be matched by the potential difference of its endpoints. When the flow of $\delta_a - \delta_b$ is realized on the edge set of $G$, both the potential difference between $a$ and $b$ and the associated effective resistance (ER) between $a$ and $b$ will be $p[a] - p[b]$; the ER of edge $e$ is $$r^\mathrm{eff}_{ab} = (\delta_a - \delta_b)^TL_G^\dagger(\delta_a - \delta_b)$$ Said in a different way, seeing a graph's edge set as a collection of resistors, if the resistors were replaced by a single resistor that acted indistinguishably from the collection subject to any flow vector - in terms of the amount of current and potential difference between edge endpoints - the amount of resistance of said single resistor is $r^\mathrm{eff}_{ab}$. Although the ER does not only summarize the relationship between nodes that are connected by an edge, for the purposes of this investigation, it is only the ERs between nodes connected by an edge that will be used directly.

Another useful property of ER in an unweighted graph is that the ER between any two nodes that are edge connected is exactly equal to the probability of those two nodes being edge connected in a random sample from the uniform spanning tree distribution of $G$. Unfortunately, to reproduce the elegant demonstration of why this is the case would require prohibitively many lines of formulae. This fact in particular provides some useful intuition as to why the ER would be a helpful tool in CD. Suppose a graph has two distinct subgraphs which are complete, but do not connect to each other except through a single edge. The ER of the edge which connects the two clusters will be very high (it will in fact be equal to one) because it is not possible for a spanning tree of the graph to exist without including this edge. The ER of an edge between two nodes in the same cluster will be comparatively small, because there are a number of spanning trees of the graph that do not make use of this edge. The ER of edges is thus a reasonable means of detecting a graph cluster's boundary because the edges which traverse that boundary have higher ERs. One last popular perspective of ER is that the expected number of steps of a random walk beginning at $v_i$, reaching $v_j$, and then ending at $v_i$ (commute distance $C_{v_iv_j}$) is proportional to the ER (with constant $2m$).

\subsection{Limitations of Effective Resistance}\label{effres_limits}

While these facets of ER provide a solid grounding for using it as a tool to understand nodes' interrelatedness, the behavior of ER described above does not necessary scale to large graphs. Specifically, it has been shown for classes of random geometric graphs - specifically kNN-graphs, $\epsilon$-graphs, and Gaussian similarity graphs - that as the size of the graph $n\rightarrow\infty$, $C_{v_iv_j} \approx \frac{1}{d_{v_i}} + \frac{1}{d_{v_j}}$, the sum of the inverse of their degrees \cite{luxburg_getting_2010}. Luxburg \textit{et al.} (2010) describe in their paper ``Getting Lost in Space" that as the graph size grows, the number of paths between any two nodes increases and a random walk will tend to ``forget" where it began, ultimately depending only on the degrees of the points at which it starts and ends. This limitation has in fact motivated a number of modified resistance-based graph distances that circumvent the scalability issue \cite{luxburg_getting_2010,nguyen_new_2016}. ER's inability to scale to large datasets carries two significant caveats involving graph types and which scenarios the convergence to $\frac{1}{d_{v_i}}+\frac{1}{d_{v_j}}$ is weakest. Regarding the first point, the convergence results are proven only for graphs with a minimal degree that grows with the graph size. This is not an unreasonable assumption in kNN- or $\epsilon$-graphs generated from data, but it is not an assumption that matches the properties of empirical networks. Secondly, the limitations of ER as a graph distance metric are only shown to be problematic for summarizing global structure, such as for nodes that are not close in the shortest path distance. CD is primarily concerned with local properties of real-world graphs with strong power law distributions. Furthermore, the ERs that will be used in this study will concern only the ER between nodes that are edge connected, i.e. very local.

\section{A TWO-STEP COMMUNITY DETECTION ALGORITHM FOR SEED SET EXPANSION} \label{alg}

While many algorithms will address an unsupervised CD problem by first establishing seeds and then propagating from those seeds, in the semi-supervised context, the seeds are provided \textit{a priori}, and propagation from those seeds is what the community detector must accomplish. The algorithm we present here takes an alternative route. We introduce what we term a \textit{germination stage}. Just as the germination of a seed is its first step departing from seed form towards becoming a full-fledged plant, the germination of seed sets in CD converts the seed set into a more larger, more reliable set of nodes, from which a community may be cultivated. In this stage, the information contained in the initial seed set is culled and expanded - not to the expected community estimate, but to a richer seed set for the second stage of propagation. Our approach can be compared in spirit to optimization methods which provide a minimization algorithm with a ``warm start." 

After this first stage we take the traditional route of propagation using personalized PageRank. Our motivation for splitting the CD procedure into these stages is based on the observation that the better the quality of the seeds used for propagation, the better the output. Basing the first stage of seed germination on ER is a deliberate choice based on the theoretical properties of ER. Additionally, empirical evidence suggests that a greedy method using a resistance-based metric will initially choose community candidates at a lower false positive rate than those that would be chosen based on the PPR score vector.

\subsection{Finding the Set with Small Effective Resistance Diameter}

While the basic procedure of germinating a seed set and then propagating its information is quite general, we outline one specific implementation of this approach using the tools of ER and PageRank. The first step is to calculate pairwise ERs for every element of $E$. This can be done in two ways. The first is the inversion of the graph Laplacian. ER for any pair of nodes $u,v\in V$ can be computed from the inverse Laplacian as $$r^\mathrm{eff}_{uv} = L^\dagger_{uu} + L^\dagger_{vv} - 2L^\dagger_{uv}$$
The second approach is to sample uniform spanning trees from the graph at hand. The ER of an edge can be estimated as the proportion of the finite sample of spanning trees wherein said edge occurs. Depending on the situation, one of these two approaches might be preferable. For sparse, diagonally dominant matrices, fast solvers exist which can find $r^\mathrm{eff}_{uv}$ in $O\left(n\mathrm{log}(n)\mathrm{log}(\epsilon^{-1})\right)$ where $n$ is the number of nonzero entries and $\epsilon$ is an accuracy level \cite{schaub_sparse_2017}. Spanning trees may be sampled (without any restriction on the type of graph) in $O(n^\frac{5}{3}m^{1}{3})$ \cite{durfee_sampling_2017}. The practitioner should thus decide which approach to use based upon the degree of sparsity, the required accuracy, and the ease of implementation. Furthermore, we are intrigued by the possibility of sampling spanning trees in a more efficient manner that guarantees an approximate estimate of the ER with high probability. The outcome of these procedures will lead to an estimate of the ER for every edge that appears in the graph (in the case of spanning trees) or for every pairwise relationship between nodes (in the case of graph Laplacian inversion). For the purposes of our algorithm, though, it is sufficient, even in the second approach, to store and use only the ERs between nodes that are connected by an edge in $G$.

A greedy algorithm is then implemented to progressively grow the seed set with the objective of minimizing the ER diameter at every stage. The specifics are found in Algorithm \ref{alg:greedy}, where $R$ denotes a sparse matrix containing ERs between nodes that are connected by an edge, $\mathrm{seeds}$ are the initial seeds provided, and $f$ is a community structure scoring function, namely conductance.

\begin{algorithm}
\caption{Greedy Minimizer of ER Diameter}\label{alg:greedy}
\begin{algorithmic}[1]
\Procedure{GreedyMinimizer}{$R,\mathrm{seeds},f$} \label{alg:greedy:name}\Comment{Produce a useful ordering of the elements of $V$ w.r.t. seed nodes}
\State $\mathrm{cur\_comm} \gets \mathrm{seeds}$
\State $\mathrm{unused} \gets V - \mathrm{seeds}$
\State $s = \{\}$
\While{$\mathrm{len}(\mathrm{unused})>0 \ \wedge \mathrm{checkStop}(s)$}
\State $i = \mathrm{argmin}_{k \in \mathrm{unused}}\left( \mathrm{min}(R[k,\mathrm{cur\_comm}])\right)$ \label{lineref1}
\State $v = \mathrm{unused}[i] $
\State $\mathrm{cur\_comm} = \mathrm{cur\_comm} \cup \{v\}$
\State $s = s \cup f(\mathrm{cur\_comm})$
\State $\mathrm{unused}= \mathrm{unused} \setminus \{v\} $
\EndWhile\label{euclidendwhile}
\State \textbf{return} $\mathrm{cur\_comm}$
\EndProcedure
\end{algorithmic}
\end{algorithm}

Two points in the algorithm must be clarified: the stopping condition, and the complexity of line \ref{lineref1}. Regarding the stopping function, we employ a procedure introduced by Yang \textit{et al.} (2015) who conduct a ``sweep" of the PageRank vector and calculate the successive conductances of larger and larger subgraphs around the seed nodes; they consider the points where local minima occur to be community candidates. Since our objective in this stage of the algorithm is to preserve the purity of the germinated seed set as well as possible, we avoid false positives by always taking the first local minimum as a stopping point in the algorithm. Another convention of Yang \textit{et al.} (2015) was to disregard any local minima for which the scoring function does not subsequently increase 20\% before arriving at another local minimum. Since our procedure for identifying a local minimum must control the false positive rate in the germination stage, we use a stricter criteria of 5\%. In any case, this means that the $\mathrm{cur\_comm}$ returned by Algorithm \ref{alg:greedy} must be trimmed back to the point at which the local minimum actually occurred and cannot be used as is. Though it has not been observed to occur in practice, if it were the case that no local minimum meets this criteria, then the entire set of nodes would be returned.

Next, regarding the time complexity for an algorithm involving a search for the argument of the minimum of a set of values each of which is the minimum of another list, we show in Theorem \ref{thm:complexity} that the time complexity is $\mathcal{O}(1)$. We assume at the onset that the graph at hand has a maximum node degree of $d$, which grows at a slower rate than graph size. Similarly, we assume that the maximum community size $n_c$ is not a graph property that grows at the same rate as graph size. These are basic features of the Lancichinetti-Fortunato-Radicchi (LFR) graph model \cite{lancichinetti_benchmarks_2009}.

\begin{theorem} \label{thm:complexity}
The time complexity of Algorithm \ref{alg:greedy} is $\mathcal{O}(1)$.
\end{theorem}
\begin{proof}
In general, finding the minimum of a list of length $n$ is of complexity $\mathcal{O}(n)$. In the worst case scenario of line \ref{lineref1} (when the graph is complete) the number of required operations would be: $$\sum_{a=1}^{n - \mathrm{len}(\mathrm{seeds})} a(n-a) = \sum_{a=1}^{n - \mathrm{len}(\mathrm{seeds})} an-a^2$$
Since $\mathrm{len}(\mathrm{seeds})$ is a small fraction of the entire community, this term may be approximated as $\sum_{a=1}^{n} an-a^2 = n \sum_{a=1}^na - \sum_{a=1}^n a^2$. In this formulation, both sums are on the order of $n^3$. In practice, we observe a local minimum to occur very early, after only considering $\approx \frac{n_c}{10}$ (though the actual proportion would depend on the parameter which describes the power law distribution of the vertices' degrees). But even if we are to consider the algorithm to run for $n_c$ steps, the total operations remain: $\sum_{i=1}^{n_c} id = d\frac{n_c(n_c+1)}{2}$. So, under our assumptions that we are in a sparse graph setting, where community sizes do not grow at the same rate as the number of nodes in the graph,
the complexity of Algorithm \ref{alg:greedy} is $\mathcal{O}(n_c^2) = \mathcal{O}(1)$.
\end{proof}

\subsection{Theoretical Properties of the Germinated Seed Set}

The objective of finding a subgraph which contains the seeds and also has the smallest ER diameter is the problem we wish to solve. There is even some recent work in graph clustering which attempts to find these subgraphs efficiently. One recent paper introduced a polynomial time algorithm to partition a graph into sets of nodes which have a bounded ER diameter \cite{alev_graph_2017}. This approach suggests an important inroad into graph clustering based on ER, but is not directly applicable to identifying a single community. It is also worth considering that the solution to the relaxed minimization problem is extremely fast and does not add much computational overhead to the use of PPR. That is, even with the computation of ERs, the germination of a seed set, and subsequent PPR, the overall complexity is still sub-quadratic. Furthermore, by demonstrating the efficacy of even an incredibly greedy algorithm for minimizing ER diameter, we aim to motivate more advanced methods for finding low ER diameter subgraphs which could be coupled with a subsequent stage of PPR.

That being said, we demonstrate that a bound exists at every step of the while loop of our algorithm. Specifically, the increase of ER diameter at each step is bounded above by the minimum of the ERs between all vertices already included and the newest addition.

\begin{theorem} \label{thm:greedy}
At each stage of the while loop in Algorithm \ref{alg:greedy:name}, the increase of the ER diameter of $\mathrm{cur\_comm}$ is bounded above by the minimum of the ERs between the newest node and the nodes in the current state of $\mathrm{cur\_comm}$.
\end{theorem}
\begin{proof}
It must be noted first that $r^\mathrm{effRes}:V\times V\rightarrow \mathbb{R}$ is a distance metric and that $(V,r^\mathrm{effRes})$ is a metric space. Suppose that at step $i$ the newest candidate for community membership to $\mathrm{cur\_comm}$ is $w\in V$. This could be the case only if $\forall u\in V, \ \mathrm{min}\left(\{r^\mathrm{effRes}(u,v) \mid v \in \mathrm{cur\_comm}\}\right) = \mathrm{min}\left(\{r^\mathrm{effRes}(w,v) \mid w \in \mathrm{unused}, \ v \in \mathrm{cur\_comm}\}\right)$. $\mathrm{diam}^\mathrm{effRes}(\mathrm{cur\_comm}\cup \{w\}) \le \mathrm{diam}^\mathrm{effRes}(\mathrm{cur\_comm}) + \mathrm{min}\left(\{r^\mathrm{effRes}(u,v) \mid v \in \mathrm{cur\_comm}\}\right)$.
\end{proof}

\section{EXPERIMENTS} \label{experiments}

\subsection{Data Description and Approach}

Three different forms of network data are considered below. The first type is a network sampled from a hierarchically-structured stochastic block model (HSBM). The matrix of probabilities along with a sample realization are depicted for ease of visualization in Figure \ref{fig:hsbm}. This is not a network structure commonly found empirically, but we use this model as a demonstration for the way an ER method works and how we expect it to work in a general network. 

\begin{figure}[h]
\vspace{.3in}
\centering
\begin{tabular}{c c}
\includegraphics[scale=0.25]{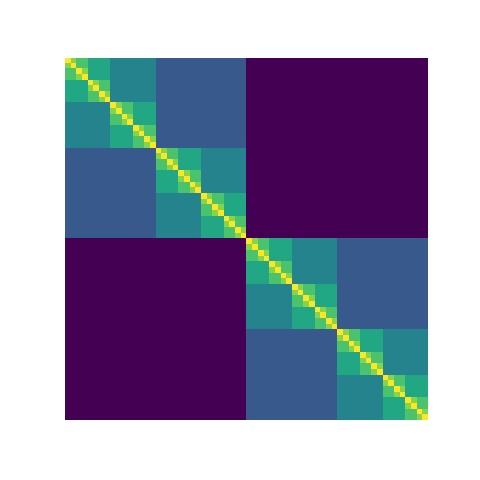} & \includegraphics[scale=0.25]{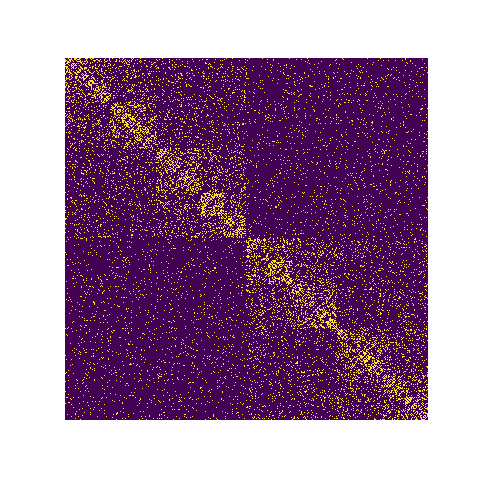}\\
\end{tabular}
\vspace{.3in}
\caption{A hierarchical SBM and its realization}\label{fig:hsbm}
\end{figure}

Next we present synthetic network data generated from the LFR graph model, a more realistic network model than the HSBM \cite{lancichinetti_benchmarks_2009}. This graph model assumes all the graph's vertices' degrees are distributed according to an exponential distribution. Depending on the parameter of this distribution, graphs may vary from a sparse ``hub and spokes" structure to a single cluster with a nearly uniform degree distribution. Using this data we show that the ER component of our method behaves differently depending on what proportion of the community is a high degree node. We also demonstrate that as long as the size of the community grows at a slower rate than the graph as a whole, the performance of our method does not degrade as a function of graph size. This is a specific concern of ours considering the limitations of using ER as a global similarity metric in large networks, as described in Section \ref{effres_limits}.

Lastly, we consider several large data benchmarks for CD. We selected some of the most commonly used networks with labeled ground-truth communities. Though these labelings are not perfect, these benchmarks elicit the closest match to real-world applications available. The specific datasets we consider are the blogs hyperlink network, Amazon co-purchasing network, the DBLP computer science co-authorship network, the Youtube social network, the Orkut social network, and web-linkages between Wikipedia pages \cite{snapnets}.

While many CD methods are validated by juxtaposing an estimated partition with the ground-truth community partition, in the one-community example considered in this paper, the objective is to maximize a combination of precision and recall for the single estimated community. We quantify our method next to traditional PPR using F1 scores. 
While the comparison of a community estimate to the ground truth is ultimately based on the F1 score or other scores calculated from precision and recall, we also provide a complete precision-recall curve for the entire sequence of the $\mathrm{cur\_comm}$ vector in order to profile whether a higher precision method is simply more conservative, or if it is uniformly outperforming a competitor. After all, the objective of our study is not specifically to produce an out-of-the-box cutting edge CD scheme, but to demonstrate how much the seed germination stage contributes before PPR begins. For this reason, our benchmark in these experiments is a vanilla PPR algorithm whose only parameter is $\alpha$, which determines how much weight is redistributed at the initial seed nodes in each iteration.

To conduct each experiment we select a community from to the ground truth labeling, and select a fixed number of seeds from that community. Then, we provide the network's edge set and its chosen seeds to the CD algorithm. When results are averaged, this reflects our having run the experiment independently on several randomly sampled synthetic networks, or our having queried a community estimate around different seeds or a different community label in each iteration.

\subsection{Hierarchical Stochastic Block Models}

The basic intuition that a network node will have lower resistance towards a high degree node than a low degree node is put to the test by this graph model. We show that in the extreme scenario where node degrees are distributed uniformly and there is sufficient within-community interconnectedness, ER will identify all community nodes reliably before making false positives. In contrast to this, we attribute ER's inability to be a standalone CD method in general to the relatively high ER between the hub of a community and its spokes relative to the resistance between the hub of one community and the hub of another. In this general case, only a few paths exist for a low degree node to the core of its community (corresponding to a high resistance path), though there may be more paths from one high degree node to another, even if they are found in distinct communities. In the HSBM setting, all the nodes' degrees are distributed uniformly in a small interval, avoiding this limitation. The ability of ER to perform and scale well in this setting (competitively with PPR), and not even be greatly improved on by a germination stage, is evident in Figure \ref{fig:hsbm:fig1}.

\begin{figure}[h]
\vspace{.3in}
\centering
\includegraphics[scale=0.5]{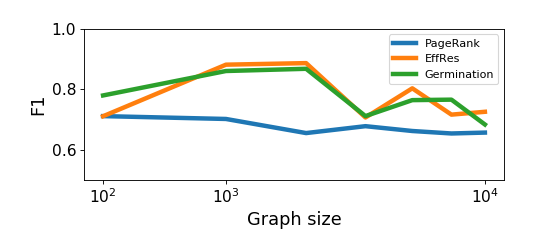}
\vspace{.3in}
\caption{F1 scores for networks based on the HSBM model as a function of graph size}\label{fig:hsbm:fig1}
\end{figure}

\subsection{Synthetic Network Data} \label{exp:synth}

The next network type used is the LFR model of benchmark graphs \cite{lancichinetti_benchmarks_2009}. In this model, the heterogeneity of the distribution of node degrees and community sizes are controlled by a power law distribution, along with a parameter governing how high a proportion of a node's edges are inward-facing versus outward-facing. In this more realistic setting, with parameters to tune governing what proportion of a network is ``popular" in its community, the contribution of ER to PPR can be assessed.

We show in this more realistic setting that pairing an ER-based germination stage with PPR outperforms PPR alone. In the top portion of Figure \ref{fig:synth} we show that this result is not dependent on the size of the graph and in the bottom portion we show that this result is not dependent on the parameter of the power law distribution of the model. Though it is apparent that graphs with fewer degrees overall pose more challenging CD tasks, the relative performance of our method is consistently superior.

\begin{figure}[h]
\vspace{.3in}
\centering
\begin{tabular}{c}
\includegraphics[scale=0.45]{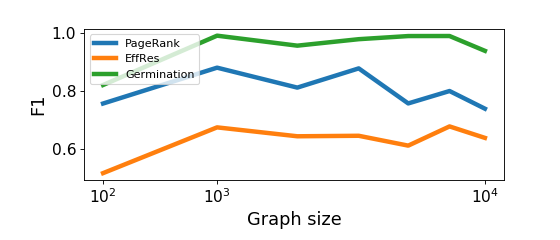} \\
\includegraphics[scale=0.45]{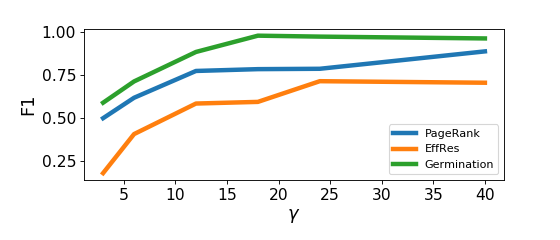} \\
\end{tabular}
\vspace{.3in}
\caption{F1 scores for networks based on the LFR model as a function of graph size (top) and power law parameter (bottom)}\label{fig:synth}
\end{figure}

\subsection{Empirical Network Data}

Lastly, we implement our CD approach on publicly available, large-scale, naturally occurring networks wherein seed set expansion is particularly relevant. To assess the general performance of these methods, we consider the mean F1 score over 10 experiments of sampling seeds from a fixed ground truth community. A summary of the (approximate) network sizes, and average F1 scores are presented in Table \ref{emp:table1}. ``Germ." in the table denotes using a germination stage pre-PPR. It is worth noting that the results are averaged over experiments of varying difficulty, which is addressed in the supplementary material. When the difference in F1 score across all experiments is considered between PPR and germination, PPR was improved upon 85\% of the time.

\begin{table}[h]
\caption{Mean of F1 scores on empirical networks} \label{emp:table1}
\begin{center}
\begin{tabular}{c c| c c }
\textbf{NAME}  & $n$ & \textbf{PPR} (sd) & \textbf{Germination} (sd)\\\hline
Blogs	&			$10^3$&	0.95 (0.015)	&	\textbf{0.96 (0.003)}		\\
Amazon         &$10^5$ &0.69 (0.07) &\textbf{0.74 (0.02)} \\
DBLP            &$10^5$ &0.58 (0.12)&\textbf{0.67 (0.13)}\\
Youtube             &$10^6$ &0.42 (0.04)&\textbf{0.49 (0.05)}\\
Orkut				& $10^6$ &			0.49 (0.21)							&		 \textbf{0.60 (0.23)}		\\
Wikipedia		& $10^6	$	& 			0.57 (0.25)				&		\textbf{0.68 (0.21)}					\\
\end{tabular}
\end{center}
\end{table}

For purposes of illustration, we take one instance of the CD task in the Amazon dataset and show the ROC curves for each CD method based on ground truth co-purchasing data. Each ROC curve depicts the precision-recall tradeoff for a progressively larger community estimate. A scoring function would determine which community estimate is selected. In Figure \ref{fig:synth:roc1} we see the primary trend in basing a PPR procedure on top of a germination stage. Specifically, we observe that the high precision of using ER alone is short lived. After attaining $\approx 25\%$ recall, the ER attains a local minimum, at which point the germination stage stops. This triggers PPR to be initiated on the ``germinated" seed set, leading to consistently better precision and recall, without regard for the point at which the stopping function is activated.

\begin{figure}[h]
\vspace{.3in}
\centering
\includegraphics[scale=0.45]{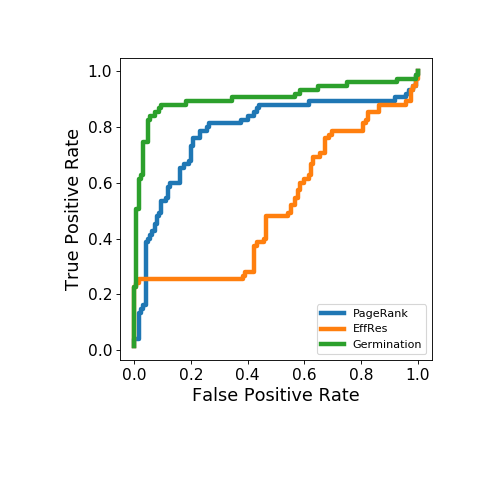}
\vspace{.3in}
\caption{ROC curve for a single run of CD}\label{fig:synth:roc1}
\end{figure}

\section{CONCLUSIONS} \label{conclusion}

We have shown that a simple modification to PPR can lead to a significant boost in accuracy for CD. While our algorithm's specific contribution to single community seed set expansion is promising, we are even more motivated by degree to which a two-step procedure is beneficial in a single community seed set expansion problem.

In the future, we plan to demonstrate further the power of this two-step approach. Possible directions include considering a larger class of community scoring functions, algorithms for locally minimizing ER diameter, and thoroughly investigating the accuracy-time tradeoffs of these decisions.

\subsubsection*{References}

\bibliographystyle{apalike}
\bibliography{ABC.bib}



\end{document}